\documentclass[a4paper,USenglish,cleveref,autoref,thm-restate]{lipics-v2021}

\usepackage[utf8]{inputenc}
\usepackage{amsmath}
\usepackage{amssymb}
\usepackage{amsthm}
\usepackage{mathrsfs}
\usepackage{graphicx}
\usepackage{hyperref}
\usepackage{afterpage}

\title{Optimizing Symbol Visibility through Displacement}

\author{Bernd Gärtner}{Department of Computer Science, ETH Zürich, Switzerland}{gaertner@inf.ethz.ch}{}{}

\author{Vishwas Kalani}{Department of Computer Science and Engineering, I.I.T. Delhi, India}{cs1200411@cse.iitd.ac.in}{}{}

\author{Meghana M. Reddy\footnote{The third author's full last name consists of two words and is \emph{Mallik Reddy}. However, she consistently refers to herself with the first word of her last name being abbreviated.}}{Department of Computer Science, ETH Zürich, Switzerland}{meghana.mreddy@inf.ethz.ch}{https://orcid.org/0000-0001-9185-1246}{Supported by the Swiss National Science Foundation within the collaborative DACH project \emph{Arrangements and Drawings} as SNSF Project 200021E-171681.}

\author{Wouter Meulemans}{Department of Mathematics and Computer Science, TU Eindhoven, the Netherlands}{w.meulemans@tue.nl}{https://orcid.org/0000-0002-4978-3400}{Partially supported by the Dutch Research Council (NWO) under project number VI.Vidi.223.137.}

\author{Bettina Speckmann}{Department of Mathematics and Computer Science, TU Eindhoven, the Netherlands}{b.speckmann@tue.nl}{https://orcid.org/0000-0002-8514-7858}{}

\author{Miloš Stojaković}{Department of Mathematics and Informatics, Faculty of Sciences, University of Novi Sad, Serbia}{milos.stojakovic@dmi.uns.ac.rs}{https://orcid.org/0000-0002-2545-8849}{Partly supported by Ministry of Science,
Technological Development and Innovation of Republic of Serbia
(Grants 451-03-66/2024-03/200125 \& 451-03-65/2024-03/200125). Partly supported by Provincial Secretariat for Higher Education and Scientific Research, Province of Vojvodina (Grant No.~142-451-2686/2021).}

\authorrunning{B. Gärtner, V. Kalani, W. Meulemans, M. M. Reddy, B. Speckmann, \& M. Stojaković} 

\Copyright{Bernd Gärtner, Vishwas Kalani, Wouter Meulemans, Meghana M. Reddy, Bettina Speckmann, and Miloš Stojaković} 



\category{} 

\relatedversion{} 

\acknowledgements{This research was initiated at the 19th Gremo’s Workshop on Open Problems (GWOP), Binn, Switzerland, June 13-17, 2022. We are grateful to the anonymous referees, whose useful and detailed comments improved our paper.}

\nolinenumbers

\hideLIPIcs

\graphicspath{{./figures/}}

\newcommand{\A}{\mathcal{A}}
\newcommand{\xx}{\mathbf{x}}
\newcommand{\yy}{\mathbf{y}}

\begin{document}
	
\maketitle

\begin{abstract}
In information visualization, the position of symbols often encodes associated data values.
When visualizing data elements with both a numerical and a categorical dimension, positioning in the categorical axis admits some flexibility.
This flexibility can be exploited to reduce symbol overlap, and thereby increase legibility.
In this paper, we initialize the algorithmic study of optimizing symbol legibility via a
limited displacement of the symbols.

Specifically, we consider closed unit square symbols that need to be placed at specified $y$-coordinates.  We optimize the drawing order of the symbols as well as their $x$-displacement, constrained within a rectangular container, to maximize the minimum visible perimeter over all squares. If the container has width and height at most $2$, there is a point that stabs all squares. In this case, we prove that a staircase layout is arbitrarily close to optimality and can be computed in $O(n\log n)$ time.
If the width is at most $2$, there is a vertical line that stabs all squares, and in this case, we design a 2-approximation algorithm (assuming fixed container height) that runs in $O(n\log n)$ time.
As it turns out that a minimum visible perimeter of 2 is always achievable with a generic construction, we measure this approximation with respect to the visible perimeter exceeding 2.  We show that, despite its simplicity, the algorithm gives asymptotically optimal results for certain instances.
\end{abstract}

\section{Introduction}
\label{sec:intro}

When communicating information visually, the position of symbols is an important visual channel to encode properties of the data. For example, in a scatter plot that visualizes age versus income of a given population, each data item (a person in the population) is visualized with a symbol (commonly a square, a cross, or a circle) which is placed at an $x$-coordinate that corresponds to their age and a $y$-coordinate that corresponds to their income. Hence persons with similar values are placed in close proximity, which allows the user to visually detect patterns. Another example from cartography are so-called proportional symbols maps which visualize numerical data associated with point locations by placing a scaled symbol (typically an opaque disc or square) at the corresponding point on the map. The size of the symbol is proportional to the data value of its location, such as the magnitude of an earthquake. The density and size of the symbols again supports visual pattern detection.

In both examples above, the position of the symbol is fixed and cannot be changed without severely distorting the information it encodes. In other settings, such as map labeling, the position of a symbol is not completely fixed, instead the symbol (the label) needs to be placed in contact with a particular point on the map. There are infinitely many potential placements for the symbol, but all must contain the same fixed point somewhere on its boundary. In this paper, we consider a related symbol placement problem, which is motivated by the visualization of numerical data with associated categories: the age of employees  within a certain division of a company or the page rank of posts that exhibit a certain sentiment (positive, negative, neutral) on a topic such as vaccinations. Such data can be visualized in \emph{categorical strips} of fixed width $w$, restricting the symbols to lie in the strip, and placing the symbols on a $y$-coordinate according to their numerical values (see Figure~\ref{fig:vac2018} for an example using Twitter data). There are again infinitely many potential placements per symbol, but all placements are restricted to share the same $y$-coordinate.

If the positions of symbols are fixed or restricted, then close symbols will overlap, reducing the visible part of -- or even fully obscuring -- other symbols. Correspondingly, there is an ample body of work on optimizing the visibility of symbols under placement restrictions. The algorithmic literature considers a couple of variants. First of all, we either display all symbols or only a subset. For symbol maps and our categorical strips we always have to display all symbols, since otherwise not all data is visualized. For map labeling one usually chooses a subset of the labels which can be placed without any overlap; the corresponding optimization problems attempt to maximize the number of these labels while also taking priorities (such as city sizes) into account. If overlap between symbols is unavoidable, visibility is optimized by either maximizing the minimum perimeter of the symbol that has least visibility or maximizing the total visible perimeter. If the positions of the symbols are completely fixed, then the only choice we can make is the drawing order of the symbols.

\begin{figure}[p]
\centering
    \includegraphics[page=1, width=\textwidth]{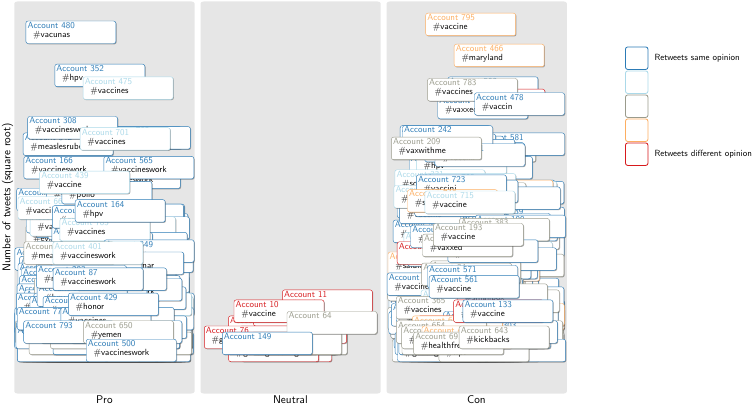} \\ \vspace{0.5cm}
    \includegraphics[page=2, width=\textwidth]{vax2018-2.pdf}
\caption{A pro and con vaccination discussion on Twitter from $2018$, capturing $823$ accounts. Accounts are placed in a categorical strip according to the most frequent sentiment expressed. Color indicates if an account mostly mentions accounts with a similar opinion (blue) or with a different opinion (red). Top: random jittering, bottom: our approximation algorithm described in Theorem~\ref{thm:main_approx}. Random jittering hides many details, such as the pro-vaccination accounts that mention predominantly anti-vaccination accounts (red boundaries). Many accounts at the bottom sent the same number of posts and hence share a $y$-coordinate, which inevitably covers most of their horizontal edges.}
\label{fig:vac2018}
\end{figure}
\afterpage{\clearpage}

In this paper, we study the novel algorithmic question of how to optimize the visibility of a given set of symbols, all of which must be drawn, when we may choose their drawing order and their $x$-coordinate, given a set of fixed (and distinct) $y$-coordinates for each symbol.
We measure the visibility of the result via the minimum visibility perimeter over all symbols~\cite{cabello2010algorithmic}.
The symbols model data of the same category to be used inside a visualization with multiple categories. Hence the $x$-coordinates are restricted to lie in a small range, to ensure visual cohesion. Specifically, we study the setting where all symbols overlap horizontally. Symbols with identical $y$-coordinates (data values) require special care, both from a visual aggregation and an algorithmic point of view. In this first study, we focus on the general case of distinct $y$-coordinates.
Figure~\ref{fig:vac2018} shows that our algorithms do indeed greatly improve the visible perimeter and thereby give the viewer a more accurate impression of the data. Note that our theoretical results hold only for square symbols, but, as evidenced by Figure~\ref{fig:vac2018}, the algorithms we propose readily extend to rectangular symbols and not necessarily distinct $y$-coordinates. Proving similar bounds for more general symbol shapes is a challenging open problem.

\subparagraph*{Contributions and organization.}
We initiate the algorithmic study of optimizing symbol visibility through displacement. Specifically, we focus on unit square symbols, that may be shifted horizontally while remaining in a strip of width 2 (their categorical strip). In Section~\ref{sec:prelims}, we introduce our notation and make some initial observations. Most notably: the visible perimeter behaves non-continuously when squares are placed on the same $x$-coordinate. Hence the optimal visible perimeter is a supremum that cannot always be reached.
In Section~\ref{sec:2x2case} we study the special case that the strip has height at most 2. In this scenario all squares are stabbed by a point. We first establish several useful geometric properties of so-called reasonable layouts -- arrangements of the input squares which meet certain lower bound conditions -- and then use these properties to prove that a simple $O(n \log n)$ algorithm suffices to compute a layout of the squares whose visible perimeter is arbitrarily close to the supremum.

In Section~\ref{sec:slabcase} we then study the general case of strips of arbitrary height (but still width~2). Here all squares are stabbed by a line.
We leverage our previous result to obtain an $O(n \log n)$-time approximation algorithm. This approximation is with respect to the so-called gap -- the visible perimeter minus two -- since a minimum visible perimeter of 2 is trivially obtained for any instance. Furthermore, if the $y$-coordinates are uniformly distributed, then we can show that a specific layout -- the zigzag layout -- is asymptotically optimal.
We close with a discussion of various avenues for future work, both towards the practical applicability of our results in visualization systems and towards more theoretical results in other settings, including different visibility definitions and other symbol shapes.

\subparagraph*{Related work.}
The questions we study in this paper combine various aspects and restrictions of well-known placement problems in the algorithmic (geo-)visualization literature in a novel way.
Most closely related to our work are the two papers by Cabello et al.~\cite{cabello2010algorithmic} and by Nivasch et al.~\cite{NivaschPT14} that consider perimeter problems for sets of differently sized symbols at fixed positions in the plane. Specifically, they consider the problems of
maximizing the perimeter of the symbol that has least visibility (as we do in this paper), or maximizing the total visible perimeter. Since the locations of the symbols are fixed, the algorithmic problem reduces to finding the optimal (not necessarily stacking) order of the symbols. This contrasts with our work where a limited form of displacement is allowed.

Labeling cartographic maps, where a subset of the labels are chosen such that they can be placed without any overlap, and the corresponding optimization problems are computationally complex in various settings~\cite{formann1991packing}, as they relate to maximum independent set. Constrained displacement of labels is often also allowed, and various of such placement models have been studied algorithmically~\cite{10.1145/3603376,depian_et_al:LIPIcs.GIScience.2023.2,Poon2004,Schwartges2014,VANKREVELD199921}.
The goal is generally to place as many labels as possible without overlap, although some variants exist that allow labels to be scaled until all labels can be placed. By design, our symbols overlap horizontally and often also vertically, and we cannot omit symbols; furthermore, scaling symbols is not an option while retaining legibility of the visualization.
Displacing labels to the boundary of a map has attracted considerable algorithmic attention under different models as well, see the paper~\cite{bekos2019external} for a survey. However, such practice relies on leader lines to connect labels to the points being labeled; making it harder to identify data patterns and thus less suitable for visualizing data.

There is ample work on using symbol displacement to eliminate all symbol overlap, as it has various applications in visualization, including visualizing disjoint glyphs in geovisualization \cite{hirono2013constrained,meulemans2019efficient,mereke2017}, removing overlap between vertices in graph drawing \cite{dwyer2005fast,marriott2003removing}, positioning nonspatial data \cite{gomez2013mixed,wordles2012} and computing Dorling and Demer's cartograms \cite{dorling,nickel2022multicriteria}.
Such overlap-removal problems are NP-hard in many settings \cite{FIALA2005306,merekethesis}, though efficiently solvable in some specific settings \cite{meulemans2019efficient,nickel2022multicriteria}. However, eliminating all overlap may require considerable displacement, thereby distorting the view of the data. Our goal is not to eliminate all overlap, but to use both a limited displacement and a suitable drawing order to maximize visibility.

In the visualization literature, offsetting graphical symbols to improve visibility is known as jittering \cite{wilke2019fundamentals}. Often, jittering is done randomly, though in context of dense plots, arising, for example, from dimensionality reduction, more complex algorithms have been engineered for this task; see e.g. \cite{giovannangeli2023guaranteed} for a recent method. However, such approaches are often heuristic in nature, without provable quality guarantees.

\section{Preliminaries}
\label{sec:prelims}

Our input is a sequence of distinct $y$-coordinates $\yy=(y_1,y_2,\ldots,y_n)$. We would like to find $x$-coordinates $\xx=(x_1,x_2,\ldots,x_n)$, determining unit squares $s_1,s_2,\ldots,s_n$, where $s_i$ has centroid $(x_i,y_i)$. We also want to find a stacking order, so that the minimum visible perimeter among all squares is maximized.

More specifically, we are given a \emph{strip} $T$ of width $1<w \leq 2$ and height $h> 1$, where we assume that $T=[0,w]\times[0,h]$, and that $\frac12\leq y_1 < y_2<\ldots <y_n\leq h-\frac12$.
Thus, in terms of the $y$-coordinates, we can speak about the \emph{highest}, or the \emph{lowest}, among any set of squares. Furthermore, we say that $s_i$ is \emph{above} (\emph{below}) $s_j$, if $y_i>y_j$ ($y_i<y_j$).
To ensure visual cohesion between the squares, we require $x_i\in [\frac12,w-\frac12]$ for all $i$, so that all squares are in the strip. We say that $s_i$ is \emph{left} (\emph{right}) of $s_j$, if $x_i<x_j$ ($x_i>x_j$). This determines a \emph{leftmost} and a \emph{rightmost} square among any set of squares (ties handled as needed).

A \emph{stacking order} is a total order $\prec$ among the squares. If $i\prec j$,
we say that $s_j$ is \emph{in front of} $s_i$, and $s_i$ is \emph{behind} $s_j$. The \emph{bottom} and \emph{top square} are the first and last squares in the order.
The pair $(\xx,\prec)$ is a \emph{layout} for the instance $(w,h,\yy)$.

The \emph{visible perimeter} of a square $s$ in a layout is the total length of all its \emph{visible} boundary, where a point on the boundary is visible if any other square $t$ containing it is behind $s$. The top square has visible perimeter $4$. The visible boundary is made up of (horizontal or vertical) \emph{visible edges}. Note that each side of $s$ can host at most one visible edge, due to all squares having the same size.

The \emph {gap} of a square in a layout is its visible perimeter minus $2$. If this is non-positive, we say that the square has no gap.
The gap of a layout is the minimum of the gaps of all squares.
This definition is motivated by the fact that we can always achieve a positive gap by suitable ``standard'' layouts which we introduce below.
Ideally, we want to find an optimal layout, one that has the maximum gap among all layouts.
However, this may not exist: one can easily construct instances where the only candidates for optimal layouts have duplicate $x$-coordinates, but no such layouts can actually be optimal.

One example of such an instance has three squares, with $y_2-y_1=2\varepsilon, y_3-y_2=\epsilon$ and a strip of width $w=1+\varepsilon$, where $\varepsilon\leq 0.2$, see~Figure~\ref{fig:no_opt}.
To see that there is no optimal layout in this case, from only looking at the two highest squares $s_2,s_3$, we see that the gap of every layout is bounded from above by $2\varepsilon$, and to achieve this gap, $s_2$ and $s_3$ together need to span the full strip width. Given this, $s_1$ cannot be the bottom square, since then it has visible perimeter at most $1+5\varepsilon\leq 2$ and hence no gap (Figure~\ref{fig:no_opt}, left); $s_2$ cannot be the bottom square, either, as this would limit its visible perimeter to at most $1+2\varepsilon$ (Figure~\ref{fig:no_opt}, middle). So $s_3$ has to be the bottom square (Figure~\ref{fig:no_opt}, right). Since $x_1=x_2$ would give the bottom square among $s_1,s_2$ a visible perimeter of at most $1+4\varepsilon$, we need to have $x_1\neq x_2$, but this means that $s_3$ has gap less than $2\varepsilon$. Hence, a gap of $2\varepsilon$ is not achievable, but any smaller gap is (by the right layout, as $x_1-x_2\rightarrow 0$).

Therefore, we take the supremum gap over all layouts as the benchmark which we want to approximate as closely as possible.

\begin{figure}[t]
     \centering
     \includegraphics[width=0.55\textwidth]{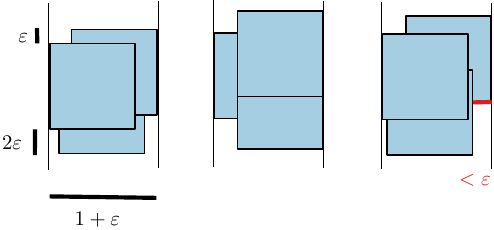}
     \caption{An instance with $y_2-y_1=2\varepsilon, y_3-y_2=\epsilon$, and a strip of width $w=1+\varepsilon$, where $\varepsilon\leq 0.2$, has no optimal layout.}
     \label{fig:no_opt}
\end{figure}

The \emph{bounding box} of a collection of squares is the inclusion-minimal axis-parallel box containing all the squares.

We call a layout a \emph{staircase} if both $\xx$ and $\prec$ are monotone:
\begin{itemize}
\item $x_1\leq x_2\leq \cdots\leq x_n$ (``facing right''), or $x_1\geq x_2\geq \cdots\geq  x_n$ (``facing left''); and
\item $s_1\prec s_2\prec \cdots \prec s_n$ (``facing up''), or $s_1\succ s_2\succ \cdots \succ s_n$ (``facing down'').
\end{itemize}
Hence, there are 4 types of staircases; the one in Figure~\ref{fig:staircase_genstaircase} (left) is facing right and up.

We call a layout a \emph{generalized staircase} if each square $s_i$ lies in one of the four corners of the bounding box of $s_i$ and all squares in front of it. In a standard staircase, this corner is the same for all squares (i.e., the lower left corner for a staircase facing right and up).

\begin{lemma}\label{obs:positive_gap}
For every instance, there is a staircase with positive gap.
\end{lemma}
\begin{proof}
We build a staircase facing up and right, with no two squares sharing the $x$-coordinate, as in Figure~\ref{fig:staircase_genstaircase} (far left). Then the left and lower sides of each square are completely visible, as well as parts of its right and upper side. This yields a positive gap.
\end{proof}

\begin{definition}
A layout is \emph{reasonable} if it has positive gap, and it is $\varepsilon$-reasonable if it has gap larger than $\varepsilon>0$.
\end{definition}

\begin{figure}[htbp]
    \centering
\includegraphics[scale=1.4,page=1]{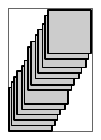}
\quad \quad \quad
        \includegraphics[scale=1.4,page=2]{staircase_genstaircase}
        \quad \quad \quad
        \includegraphics[scale=1.4,page=3]{staircase_genstaircase}
        \quad \quad \quad
        \includegraphics[scale=1.4,page=4]{staircase_genstaircase}
        \caption{A staircase (far left), and three generalized staircases.}
        \label{fig:staircase_genstaircase}
\end{figure}

\section{Squares stabbed by a point}
\label{sec:2x2case}

Throughout this section, we fix a strip width $w\leq 2$ and a strip height $h\leq 2$. In this case, all squares are stabbed by a single point.
Subsection~\ref{sec:reasonable} proves a (tight) upper bound on the gap of every layout, while Subsection~\ref{sec:staircase_is_opt} shows that a staircase layout of gap arbitrarily close to the supremum can efficiently be computed.

\subsection{Reasonable layouts}\label{sec:reasonable}
We start with a crucial structural result about reasonable layouts in the case $w,h\leq 2$.
\begin{proposition} \label{prop:bottom_must_be_leftright}
In a reasonable layout, the bottom square $s$ is not contained in the bounding box of the other squares.
\end{proposition}
\begin{proof}
  Assume otherwise and consider some two squares $s_{\ell}$ and $s_r$ defining, respectively, the left and right sides of the bounding box. If one of them is above $s$ and the other one below, the situation is as in Figure~\ref{fig:bottom_mid_square} (left). Since $s_{\ell}$ and $s_r$ overlap in both $x$- and $y$-coordinate, $s$ has horizontal and vertical visible edges of total length at most $1$ each. Thus, the visible perimeter of $s$ is at most 2. In the other case, $s_{\ell}$ and $s_r$ are w.l.o.g.\ both above $s$ as in  Figure~\ref{fig:bottom_mid_square} (right). Then we consider the square $s_d$ defining the bottom side of the bounding box; w.l.o.g. $s_d$ is left of $s$. In this case, $s_r$ and $s_d$ prove that $s$ has no gap.
\end{proof}

\begin{figure}[htbp]
  \centering
  \includegraphics[scale=1.1]{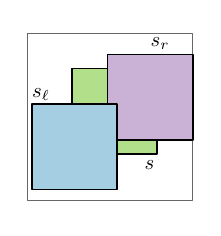}\quad\quad
  \includegraphics[scale=1.1]{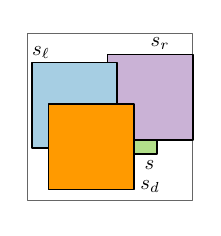}
        \caption{The visible perimeter of the bottom square $s$ is at most two if it is contained in the bounding box of the other squares.}
        \label{fig:bottom_mid_square}
  \end{figure}

\begin{lemma} \label{lem:staircase_value}
Every layout $\A$ of $n$ squares has gap at most $\frac{w+h-2}{n-1}$, for $n\geq 2$.
\end{lemma}
\begin{proof}
  If $\A$ is unreasonable, there is nothing to prove. Otherwise let $t_1,t_2,\ldots,t_n$ be the sequence of squares in stacking order, i.e.~$t_1\prec t_2\prec \cdots\prec t_n$, and let the bounding box of $t_i,t_{i+1},\ldots,t_n$ be denoted by $\tau_i$ for $1 \leq i \leq n$.
  Since $\A$ is reasonable, $t_i$ ``sticks out'' of $\tau_{i+1}$ (i.e.~it is not contained in it)
  by  Proposition~\ref{prop:bottom_must_be_leftright}, which can be applied as a sublayout of a reasonable layout is clearly reasonable as well.

  There are two cases: $t_i$ is a  ``corner square'' (Figure~\ref{fig:gen_stair_bound} left), or a ``side square'' (Figure~\ref{fig:gen_stair_bound} middle and right). Let $\Delta X_i$ and $\Delta Y_i$ quantify by how much $t_i$ sticks out, horizontally and vertically. For a side square, one of those numbers is $0$.

For a corner square, the two sides of $t_i$ incident to the corner contribute visible perimeter $2$, meaning that the gap of the square is $\Delta X_i+\Delta Y_i$. For a horizontal side square as in Figure~\ref{fig:gen_stair_bound} (middle), the visible perimeter is at most $2+\Delta Y_i$, and for a vertical side square (right), it is at most $2+\Delta X_i$. In both cases, $\Delta X_i+\Delta Y_i$ is also an upper bound for the gap of the square.

  \begin{figure}[htbp]
  \centering
      \includegraphics[scale=0.90]{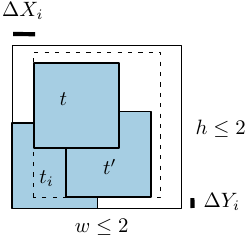}\quad
      \includegraphics[scale=0.90]{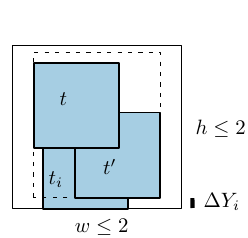}\quad
      \includegraphics[scale=0.90]{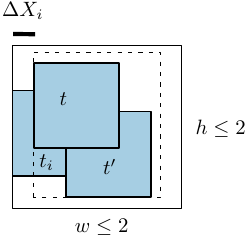}
    \caption{Proof of Lemma~\ref{lem:staircase_value}\label{fig:gen_stair_bound}.}
  \end{figure}

  This means that, for $1\leq i<n$, the (open) intervals corresponding to $\Delta X_i$ (as well as those corresponding to $\Delta Y_i$) span disjoint $x$-intervals (or $y$-intervals) that are also disjoint from the two intervals (one in each coordinate) of length one that is spanned by the top square. Hence,
  \[
    \sum_{i=1}^{n-1} (\Delta X_i+\Delta Y_i) \leq (w-1)+(h-1) = w+h-2.
  \]
It follows that there is some $t_i$ with gap at most $\Delta X_i+\Delta Y_i\leq \frac{w+h-2}{n-1}$.
\end{proof}
This upper bound on the gap is easily seen to be tight.
\begin{lemma}\label{obs:staircase_value}
There are instances of $n$ squares for which a staircase has gap $\frac{w+h-2}{n-1}$.
\end{lemma}
\begin{proof}
We consider the uniformly spaced instance ($y_i=\frac12+(h-1)\frac{i-1}{n-1}$). Choosing $x_i=\frac12+(w-1)\frac{i-1}{n-1}$ and $s_1\prec s_2\cdots\prec s_n$ leads to a staircase with $\Delta X_i=\frac{w-1}{n-1}$ and $\Delta Y_i=\frac{h-1}{n-1}$ for $1\leq i<n$, and hence the gap is  $\frac{w+h-2}{n-1}$.
\end{proof}

\subsection{Computing staircases with gap arbitrarily close to the supremum}\label{sec:staircase_is_opt}

We next prove that for every reasonable layout with gap $\gamma$, there is a staircase with a gap at least $\gamma-\delta$, for any $\delta>0$. Moreover,  with $\gamma^\star$ being the supremum gap over all layouts, a staircase of gap $\gamma^\star-\delta$ can be efficiently computed.

For this, we first look at staircases in more detail. Consider a staircase of $n$ squares, facing up and right, with centroids $(x_i, y_i)$, $1 \leq i \leq n$. We define $\Delta y_i = y_{i+1} - y_{i}>0$ and $\Delta x_i = x_{i+1} - x_{i}\geq 0$, for $1 \leq i \leq n-1$. If $\Delta x_i>0$, the left and lower sides of $s_i$ are fully visible, and the gap of $s_i$ is $\Delta y_i + \Delta x_i$; the top square $s_n$ has gap $2$. If all $\Delta x_i$ are positive, the staircase is called \emph{proper}.

Now consider the problem of finding such a proper staircase of large gap. For this, the $\Delta y_i$ are fixed, but $x_1<x_2<\cdots < x_n$ can be chosen freely, meaning that the values $\Delta x_i$ can be any positive numbers satisfying $\sum_{i=1}^{n-1} \Delta x_i \leq w-1$.
We want to maximize the gap $\min_{i=1}^{n-1}(\Delta y_i + \Delta x_i)$ subject to the previously mentioned constraints. Due to the strict inequalities on the $\Delta x_i$'s, the maximum may not exist (as pointed out in Section~\ref{sec:prelims}). But allowing $\Delta x_i\geq 0$, the maximum is attained by the solution of a linear program in the variables $\Delta x_i$ and $g$ (the $\Delta y_i=y_{i+1}-y_i$ are constants):
\begin{equation}\label{eq:lp}
  \begin{array}{lrcl}
    \text{maximize~} g \\
    \text{subject to} &  \Delta x_i + \Delta y_i &\geq& g, \quad i=1,\ldots,n-1\\
                    & \sum_{i=1}^{n-1} \Delta x_i&\leq&  w-1\\
    & \Delta x_i&\geq & 0, \quad i=1,\ldots,n-1.
  \end{array}
\end{equation}

We formalize this in the main result of this section which also implies that the optimal solution $g^\star$ of this linear program equals $\gamma^{\star}$, the supremum gap over all layouts.

\begin{lemma} \label{lem:free_y_opt_value}
  Let $\A$ be a reasonable layout with gap $\gamma$. There is a
  feasible solution of the linear program (\ref{eq:lp}) with value $g\geq\gamma$. Vice versa, for every feasible solution of (\ref{eq:lp}) with value $g$, and for every $\delta>0$, there exists a proper staircase $\A'$ with gap $g-\delta$.
\end{lemma}

\begin{proof}
  As in the proof of Lemma~\ref{lem:staircase_value}, we consider the $\Delta X_i$ and $\Delta Y_i$, $i=1,\ldots,n-1$, quantifying by how much the $i$-th square in the stacking order sticks out of the bounding box of the squares in front of it, horizontally and vertically. See Figure~\ref{fig:free_y_opt_value} (top left part) for an illustration, where the $\Delta X_i$ and $\Delta Y_i$ values are visualized as rectangle areas. Each rectangle is spanned by a unit side and an \emph{interval} corresponding to the value. For example, $\Delta X_4$, the rectangle with the orange boundary, has the interval between the left sides of $t_4$ and $t_5$.

  \begin{figure}[b]
    \centering
    \includegraphics[width=\textwidth,page=3]{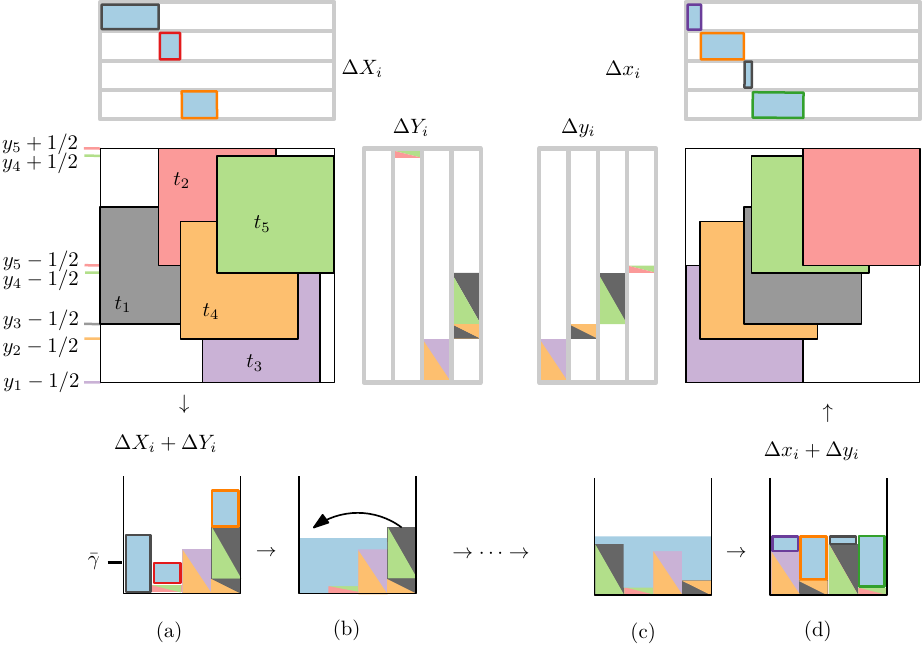}
    \caption{Proof of Lemma~\ref{lem:free_y_opt_value}.}
    \label{fig:free_y_opt_value}
  \end{figure}

  The $\Delta Y_i$ rectangles are further subdivided into \emph{blocks} whose intervals are gaps between vertically adjacent squares. For example, $\Delta Y_4$ is made of two blocks. The interval of the orange-black block is the vertical gap between the lower sides of $t_4$ (orange) at height $y_2-1/2$ and $t_1$ (black) at height $y_3-1/2$; the interval of the black-green block is the vertical gap between the lower sides of $t_1$ (black) at height $y_3-1/2$ and $t_5$ (green) at height $y_4-1/2$. For the green-red block, the interval corresponds to the gap between the  \emph{upper} sides of $t_5$ (green) at height $y_4+1/2$ and $t_2$ (red) at height $y_5+1/2$. In general, each block interval is therefore of the form $\Delta y_j= y_{j+1}-y_j$.

  As argued in the proof of Lemma~\ref{lem:staircase_value}, the
$\Delta X_i$ and $\Delta Y_i$ satisfy the inequalities
\begin{equation}
\label{eq:XYdelta}
    \Delta X_i+\Delta Y_i\geq \gamma, \quad i=1,\ldots,n-1,
\end{equation}
  with $\gamma$ being the gap of $\A$.

  We will perform discrete steps that gradually turn the $\Delta Y_i$ into the prescribed $\Delta y_i=y_{i+1}-y_i$, while changing the $\Delta X_i$ into suitable nonnegative $\Delta x_i$, maintaining their sum. If we can also maintain the constraints (\ref{eq:XYdelta}) throughout, we will arrive at a feasible solution of (\ref{eq:lp}) with value $g\geq \gamma$.

  Initially, the $\Delta X_i,\Delta Y_i$ are sorted by stacking order; our step-wise process will therefore first result in values $\Delta X'_i,\Delta Y'_i$ such that the $\Delta Y'_i$ are a permutation of the $\Delta y_i$. Without affecting the constraints (\ref{eq:XYdelta}), we finally
  reshuffle the pairs $(\Delta X'_i, \Delta Y_i)$ to get a solution of (\ref{eq:lp}).

Here is how the step-wise process works.
If the $\Delta Y_i$ are already a permutation of the $\Delta y_i$, we are done after reshuffling. This is the case if and only if each  $\Delta Y_i$ consists of exactly one block.

But in general, some $\Delta Y_i$ may have more than one block, or no block at all. In the example in Figure~\ref{fig:free_y_opt_value}, we have $\Delta Y_4=y_4-y_2$ consisting of two blocks with intervals between $t_4$ (orange) and $t_1$ (black), and between $t_1$ (black) and $t_5$ (green). $\Delta Y_1$ in turn has no blocks, as $t_1$ does not stick out vertically.

  All the $\Delta Y_i$ together use all the $n-1$ blocks. Indeed, the bounding box of the squares in front of $t_i$ is disjoint from whatever sticks out of it, and the last bounding box only contains the top square.

Now we repeatedly move blocks from rectangles with at least two blocks to rectangles with no block. We can visually analyze this as follows: think of a basin that initially holds the $\Delta Y_i$ rectangles, $i=1,\ldots,n-1$, as in part (a) of Figure~\ref{fig:free_y_opt_value}.
For each $i$, we pour $\Delta X_i$ units of water into the basin, which corresponds to the area of the rectangles with colored boundary in part  (a) of Figure~\ref{fig:free_y_opt_value}.
As $\bar{\gamma}:=\min_i (\Delta X_i+\Delta Y_i)\geq \gamma$, the water will settle at level at least $\bar{\gamma}\geq \gamma$; see part (b) of the figure.
Moving a block to a ``free slot'' will submerge it further, and this can only increase the water level; see step (b)-(c).

In the end, we have one block $\Delta y_j=y_{j+1}-y_j$ per slot, and sorting the slots by index $j$, as in part (d), yields rectangles $\Delta y_1,\ldots,\Delta y_{n-1}$, with columns $\Delta x_1,\ldots \Delta x_{n-1}$  of water above them, such that $\min (\Delta x_i+\Delta y_i)\geq \bar{\gamma}$. In general, some $\Delta y_i$'s can still be above the water level in which case the corresponding $\Delta x_i$ is $0$. This is our desired solution of (\ref{eq:lp}).

For the ``vice versa'' statement of the lemma, consider any feasible solution of (\ref{eq:lp}) with value $g$. From this, we can construct a proper staircase with gap at least $g-\delta$, by slightly redistributing the $\Delta x_i$ to make all of them positive. Then, we can in turn build a staircase with the prescribed $x$- and $y$-gaps as in the upper right part of the figure.
\end{proof}

We remark that the linear program~(\ref{eq:lp}) can be efficiently solved in $O(n\log n)$ time, employing the water analogy. After sorting the $\Delta y_i$ rectangles, and assuming that the water currently rises to the top of one of them, it is easy to compute in $O(1)$ time the amount of additional water required to reach the top of the next higher rectangle. Indeed, in this range, the water level is a linear function of the amount of additional water. If reaching the top of the next higher rectangle would need more water than our total budget of $w-1$ allows, we arrive at the optimal level $g^\star$ before.

\section{Squares stabbed by a vertical line}
\label{sec:slabcase}

Throughout this section, we consider a strip of width $w\leq 2$ and arbitrary height $h>1$, with $n$ squares of fixed $y$-coordinates $\frac12\leq y_1< y_2<\cdots< y_n\leq h-\frac12$.
Let $1/k = (h-1)/(n-1)$ be the (maximal) average $y$-distance between adjacent centroids in the $y$-order. We first show that we can asymptotically approximate the supremum gap up to a factor of $2$. More precisely, as the strip remains fixed and $n\rightarrow\infty$, we have $1/k\rightarrow 0$ and thus approach a factor of $2$ using Theorem~\ref{thm:main_approx} below. We still present our results in terms of $k$ to make it clear what happens if the strip height $h$ grows with $n$.

\begin{theorem}\label{thm:main_approx}
  Let $\gamma^\star$ be the supremum gap over all layouts. In time $O(n\log n)$, we can construct a layout with gap at least $\gamma^\star(\frac12-O(\frac1k))$. We refer to this procedure as the \emph{squeezing algorithm}.
\end{theorem}

\begin{proof}
We partition the squares into buckets $1,\ldots,\lceil h \rceil$, where bucket $i$ contains the squares $j$ such that $y_j$ rounds to $i$ (we round up in case of a tie). The squares within each bucket are in a strip of height $2$, and by Section~\ref{sec:2x2case}, a (staircase) solution of gap arbitrarily close to the supremum can efficiently be found, in time  $O(\ell \log \ell)$ per bucket, where $\ell$ is the number of squares in that bucket. Hence, the total time required is $O(n \log n)$.

The smallest bucket gap $\delta$ is (up to arbitrarily small error) an upper bound for $\gamma^\star$, as each layout contains a sublayout for the squares in this worst bucket. We also note that $\delta=O(\frac1k)$, since there must be a bucket with $\Omega(k)$ squares to which Lemma~\ref{lem:staircase_value} applies.

In $O(n)$ time, we now construct a layout for all squares, of gap roughly $\frac\delta2$, to prove the statement. To do so, we ``squeeze'' the layouts in individual buckets appropriately.  We assume w.l.o.g.\ that the even bucket staircases are facing up and right, while the odd ones are facing up and left, as in Figure~\ref{fig:ne_nw} (left).

\begin{figure}[htbp]
    \begin{minipage}[t]{.49\textwidth}
        \centering
        \includegraphics[scale=1]{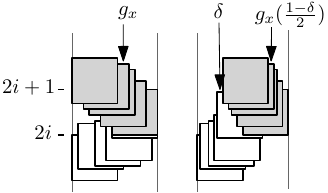}
        \caption{Squeezing staircase layouts.}
        \label{fig:ne_nw}
    \end{minipage}
    \hfill
    \begin{minipage}[t]{.49\textwidth}
        \centering
        \includegraphics[scale=1]{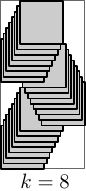}
        \caption{Arranging $n$ uniformly spaced squares in a zigzag layout.}
        \label{fig:uniform_placing}
    \end{minipage}
\end{figure}

Multiplying all $x$-gaps by $\frac{1-\delta}{2}$ while keeping the even staircases aligned left and the odd ones aligned right, see Figure~\ref{fig:ne_nw} (left), leads to a layout where even staircase squares have $x\leq 1-\frac\delta2$, and odd ones have $x\geq 1+\frac\delta2$. Each non-top square of each bucket still has gap $g_y+g_x\frac{1-\delta}{2}$ where $g_x,g_y$ are the previous $x$-gap and $y$-gap in the bucket solution, and $g=g_x+g_y\geq \delta$ is the previous gap. It follows that the new gap is at least $\delta\frac{1-\delta}{2}$. The top squares of each bucket have $x$-gap (and hence total gap) at least $\delta$, by construction.
The resulting layout has therefore gap at least $\delta\frac{1-\delta}{2}$. Since $\gamma^\star \lessapprox \delta=O(1/k)$, the bound follows.
\end{proof}

It is natural to ask whether squeezing the staircase layouts of individual buckets is the best we can do. For general $y_i$, we do not know the answer, but if the $y_i$ are uniformly spaced, we can indeed prove that this procedure yields an asymptotically optimal gap.

\subparagraph*{Uniform spacing.} For the rest of the section, we assume that $y_{i+1}-y_i=\frac1k$ for $1\leq i<n$.  In this case, the squeezing algorithm from Theorem~\ref{thm:main_approx} essentially produces the \emph{zigzag} layout (see Figure~\ref{fig:uniform_placing}).

\begin{lemma}
  The zigzag layout has gap
  $
    \frac1k+\frac{1}{2k-1}.
  $
\end{lemma}

\begin{proof}
  We place bundles of $\lfloor k \rfloor$
  squares each, as indicated in Figure~\ref{fig:uniform_placing}, starting from the lowest one. This layout uses precisely the $2\lfloor k\rfloor$ $x$-coordinates
  $\frac12 + \frac{i}{2\lfloor k \rfloor-1}, i=0,\ldots, 2\lfloor k\rfloor-1$. This means that every square has $x$-gap at
  least $\frac{1}{2\lfloor k \rfloor-1} \geq \frac{1}{2k-1}$. The  $y$-gap is at least $1/k$ for
  each square, due to uniform spacing. Both gaps are attained for
  example by the second-lowest square, so the bound in the lemma cannot
  be improved for this layout.
\end{proof}

Below, we will establish the following result, showing that the simple zigzag layout is asymptotically optimal.
\begin{theorem}\label{thm:main_uniform}
  In the case of uniform spacing, every layout has gap at most
  $
    \frac1k+\frac{1}{2k-O(\log k)}.
  $
\end{theorem}

In proving this, we can restrict to $1/k$-reasonable layouts, the ones achieving gap larger than $1/k$ in the first place. We also assume that $k\geq 2$.

We will start by establishing a crucial fact about such layouts, namely that most of their squares have 3 visible corners. To this end, we are going to upper-bound the number of squares with at least 2 covered corners, eventually enabling us to remove them from the layout while keeping most of the squares.

\begin{definition}
  Given a layout, a \emph{bad square} is one with at least 2 covered corners. A bad square with one vertical side covered is a \emph{standard bad square}; see Figure~\ref{fig:Case2left} (left).
\end{definition}

\begin{lemma}\label{lem:ad_sq}
If a square $s$ has both adjacent squares (in the $y$-order) in front of it, then $s$ is a standard bad square.
\end{lemma}

\begin{proof}
  If the adjacent squares both have smaller or larger $x$-coordinate, then they together hide a vertical side of $s$, see Figure~\ref{fig:Case2left} (b). The other case cannot happen in a reasonable layout by Proposition~\ref{prop:bottom_must_be_leftright}; see Figure~\ref{fig:Case2left} (c)-(d). Since  $k\geq 2$,  the adjacent squares actually overlap vertically.
\end{proof}

\begin{figure}[htbp]
\centering
\includegraphics[page=2]{Case2Left.pdf}
\caption{Bad squares: at least two covered corners; A standard bad square ((a) and (b)): one vertical side is covered. \label{fig:Case2left}}
\end{figure}

Counting standard bad squares yields a bound for all bad squares.

\begin{lemma}
For each non-standard bad square, an adjacent square (in the $y$-order) is a standard bad square.
\end{lemma}

\begin{proof}
Let $s$ be a non-standard bad square. We distinguish two cases.

The first one is that an upper corner and a lower corner of $s$ are covered. These could be adjacent corners (with some part of the connecting side visible), or antipodal corners as in Figure~\ref{fig:Case3left}. In both of those instances, by Lemma~\ref{lem:ad_sq}, one of the adjacent squares must be behind $s$; w.l.o.g.\ it is the next higher one $b$ (blue).

\begin{figure}[htbp]
    \begin{minipage}[t]{.57\textwidth}
        \centering
        \includegraphics[scale=1]{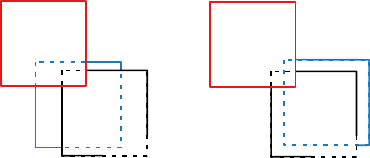}
        \caption{Case 1: A non-standard bad square (black) with an upper and a lower corner covered.}
        \label{fig:Case3left}
    \end{minipage}
    \hfill
    \begin{minipage}[t]{.39\textwidth}
        \centering
        \includegraphics[scale=1]{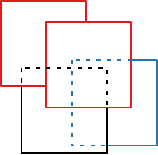}
        \caption{Case 2: A non-standard bad square (black) with two upper or two lower corners covered.}
        \label{fig:Case5left}
    \end{minipage}
\end{figure}

Consider the square $r$ (red) covering the upper corner. Square $b$ is behind $s$ and $r$, and either ``wedged'' between them (w.r.t.~to both $x$- and $y$-coordinate), or ``sticking'' out.
The former case (Figure~\ref{fig:Case3left} left) cannot happen, because $b$ would have no gap then, see Proposition~\ref{prop:bottom_must_be_leftright}. In the latter case, $b$ is the required standard bad square (Figure~\ref{fig:Case3left} right). This uses that $r$ is higher than $b$ due to uniform spacing.

The second case is that two upper or two lower corners of $s$ are covered, see  Figure~\ref{fig:Case5left}. Let us suppose w.l.o.g.\ that the two upper corners are covered. Then the upper side of $s$ is covered. This implies that the next higher square $b$ is behind $s$, as otherwise, $s$ has gap at most $1/k$. Again, $b$ is a standard bad square.
\end{proof}

Through the previous lemma, each standard bad square is ``charged'' by at most three bad squares (itself and the two adjacent ones).

\begin{corollary}\label{cor:standard_bad}
For every vertical window $W=[\underline{h},\overline{h}]\subseteq[\frac12,h-\frac12]$ of a $1/k$-reasonable layout, the number of bad squares with $y_i\in W$ is at most three times the number of standard bad squares with $y_i\in W'=[\underline{h}-\frac1k,\overline{h}+\frac1k]$.
\end{corollary}

It remains to count the number of standard bad squares.

\begin{lemma}\label{lem:few_type_2}
  For every vertical window $W=[\underline{h},\underline{h}+1]$ of a $1/k$-reasonable layout, there are at most $2(\log k+1)$ standard bad squares with
$y_i\in W$.
\end{lemma}

\begin{proof}
  Let us fix the window. We count the standard bad squares with the left side covered, the overall bound follows by symmetry.

  Let $s_1,\ldots,s_{\ell}$ be these squares; see Figure~\ref{fig:standard_bad_squares} (left). They must be stacked according to $x$-coordinate, with squares of lower $x$-coordinate in front of squares with higher $x$-coordinate. Indeed, a square $s_i$ in front of a square $s_j$ with smaller $x$-coordinate would cover a third corner of $s_j$, and thus a full horizontal side, resulting in no gap (we are using here that the window height is $1$). The squares covering the left side of $s_i$ are to the left of $s_i$ and together cover all of $s_i$ in the left half of the strip.

\begin{figure}[htbp]
\centering
    \includegraphics{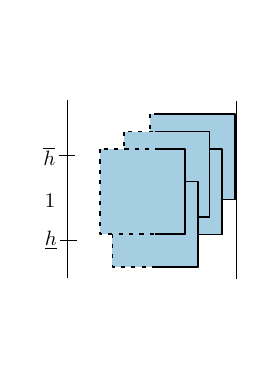} \quad \includegraphics{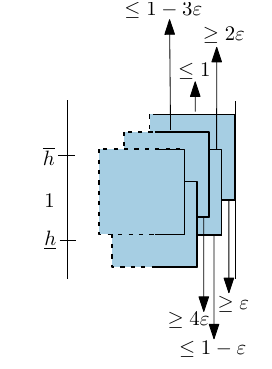}
\caption{Counting standard bad squares with centers in a vertical window of height $1$\label{fig:standard_bad_squares}.}
\end{figure}

Because the layout is $1/k$-reasonable, each $s_i$ has a part of each of its horizontal sides visible. They are of lengths $\sigma_i\leq \lambda_i\leq 1$ such that $\sigma_i+\lambda_i\geq 1 + 1/k$.  Suppose that the squares are ordered by decreasing $x$-coordinate. We show that the $\sigma_i$ increase exponentially with $i$.

We have $\lambda_1\leq 1$, hence $\sigma_1\geq \varepsilon :=1/k$; see Figure~\ref{fig:standard_bad_squares} (right). As a consequence, $\lambda_2\leq 1-\varepsilon$ (as $s_2$ is by at least $\varepsilon$ further to the left than $s_1$). Hence, $\sigma_2\geq 2\varepsilon$. This in turn means that $s_3$ is by at least $\epsilon+2\varepsilon$ further to the left than $s_1$, so $\lambda_3\leq 1-3\varepsilon$ and $\sigma_3\geq 4\varepsilon$.

Continuing in this fashion, we see that $\sigma_{\ell}\geq 2^{\ell-1}\varepsilon\leq 1$. This implies that $\ell-1\leq \log(1/\varepsilon)=\log k$.
\end{proof}

\begin{corollary}
In a reasonable layout, at most $6(\log k +1)$ squares out of any consecutive $k-1$ squares are bad squares.
\end{corollary}

\begin{proof}
  The centers of $k-1$ consecutive squares span a horizontal window of height $1-2/k$. Using Corollary~\ref{cor:standard_bad} and the previous lemma, the number of bad squares in this window is at most $3(2(\log k+1))$.
\end{proof}

Hence, by removing $O(\log k)$ squares per bundle of $k-1$ squares, we obtain a layout with no bad squares left (observe that no surviving square can turn bad by removing squares).
Such a layout turns out to have a rather rigid structure.

\begin{lemma}\label{lem:monotone} After removal of all bad squares from a $1/k$-reasonable layout, there is a unique top square (fully visible), and the stacking order is determined: monotone decreasing from the top square towards the highest as well as the lowest square.
\end{lemma}
\begin{proof}
  A square with at least three visible corners (and only such squares remain) is called a \emph{down square} if the lower side is fully visible, and an \emph{up square} if the upper side is fully visible. A top square is both up and down.

  Now let the squares be indexed from lowest to highest. We claim that if $s_i$ is a down square, then $s_{i-1}$ is also a down square that is behind $s_i$. To see this, consider a down square $s_i$ and the overlapping square $s_{i-1}$ (we have an overlap since we have removed only $O(\log k)$ squares in between); see Figure~\ref{fig:top_down}.
It is clear that $s_{i-1}$ must be behind $s_i$, and this in turn implies that $s_i$ is also a down square.

\begin{figure}[b]
    \begin{minipage}[t]{.34\textwidth}
        \centering
        \includegraphics[scale=1]{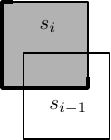}
        \caption{Proof of the bitone stacking order. The visible perimeter of $s_i$ is marked in bold.}
        \label{fig:top_down}
    \end{minipage}
    \hfill
    \begin{minipage}[t]{.62\textwidth}
        \centering
        \includegraphics[scale=1.5,page=5]{staircase_genstaircase} \quad \includegraphics[scale=1.5,page=6]{staircase_genstaircase}\quad
        \includegraphics[scale=1.2,page=2]{uniform_proof}
        \caption{Two squares of different types in the subwindow.}
        \label{fig:uniform_proof}
    \end{minipage}
\end{figure}

 A symmetric statement holds for up squares. Hence, starting from \emph{any} top square, we can go both higher and lower, decreasing stacking height. In particular, we can never encounter another top square.
\end{proof}

We now proceed to the proof of Theorem~\ref{thm:main_uniform}, showing that, under uniform spacing, we cannot asymptotically beat the gap of the zigzag layout.

\begin{proof}[Proof of Theorem~\ref{thm:main_uniform}]
  We start with the layout obtained after removing all bad squares, as in Lemma~\ref{lem:monotone}. W.l.o.g.\ we assume that the majority of squares is below the top square, and we disregard all squares above. The stacking order then coincides with the $y$-order.

We now consider a window of height $4$. Each square with center in that window is either left or right of the next higher square, and based on this, we call it type L, or type R. We focus on the middle subwindow of height $2$. If all squares with centers in this subwindow have the same type, we have a proper staircase of $2k-O(\log k)$ squares; see Figure~\ref{fig:uniform_proof} (left).

Except $O(\log k)$ of them (the ones directly below a removed bad square), all squares have $y$-gap $1/k$. Let $\varepsilon$ be the minimum $x$-gap among these squares. Then the layout has gap at most $1/k+\varepsilon$. But we know that the sum of all $x$-gaps is at most $1$ (they do not overlap and ``live'' outside the top square), so the minimum $x$-gap is  $1/(2k-O(\log k))$ which proves the theorem in this case.

The other case is that there are two consecutive squares of different types with centers in the subwindow, see Figure~\ref{fig:uniform_proof} (middle). In this case, we have a generalized staircase. We let $s_{i}$ be the lower one and $s_{i'}$ the higher of the two squares, and we zoom in on the situation; see Figure~\ref{fig:uniform_proof} (right).

Both $s_i$ and $s_{i'}$ stick out of the squares $s_j$, $j>i'$, that are up to $1$ higher than $s_i$ (otherwise, they cannot have $3$ visible corners). On the other hand, the squares $s_j$, $j<i$, that are up to $1$ lower than $s_{i'}$ stick out of both $s_i$ and $s_{i'}$. It follows that---as in the previous case---the $x$-gaps of all involved squares live outside of the top square among them, and they do not overlap (i.e.~they are disjoint). More specifically, the $x$-gaps of the squares above $s_i$ live in the orange regions in Figure~\ref{fig:uniform_proof}, while the $x$-gaps of the squares below $s_{i'}$ live in the blue regions.

In total, we again have $2k-O(\log k)$ squares with a $y$-gap of $1/k$ within the surrounding window of height $4$, so the minimum $x$-gap is $1/(2k-O(\log k))$.
\end{proof}

\section{Conclusion}
\label{sec:conclusion}

We initiated the algorithmic study of optimizing the visibility of overlapping symbols by finding both a suitable drawing order and a limited displacement. This novel setting leads to various interesting and challenging problems. In this paper, we focused solely on unit squares, presented structural insights, as well as several intricate approximation algorithms.

We are curious if the upper bound from Theorem~\ref{thm:main_uniform} can be improved to one where the $O(\log k)$ term is replaced with a constant. In our approach we derive the bound by eliminating all the \emph{bad} squares from the layout before estimating the gap. Hence, knowing that (in a vertical window of mutually intersecting squares) the number of bad squares can be as large as $\Theta(\log k)$ a radically new approach would be needed to achieve such bound improvement.

It is natural to wonder if the stacking order of every optimal solution follows its $y$-order.
However, this is not always the case,
refer to Figure~\ref{fig:difference}: the optimal layout (leftmost figure) is better than the best layout when the stacking order follows the $y$-order or the inverse $y$-order (second and third figures), which is better than the best layout among all staircases (rightmost figure).

\begin{figure}[t]
  \centering
    \includegraphics[width=\textwidth,page=1]{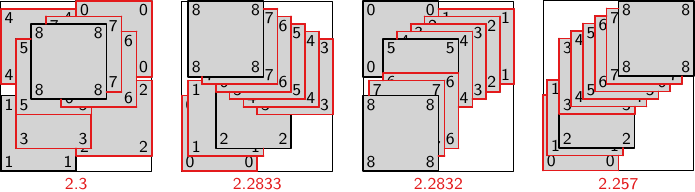}
\caption{Four layouts for $y = \{0.5, 0.7, 0.8, 1.25, 1.35, 1.45, 1.55, 1.65, 1.75\}$ in a strip of width~2. The minimum visible perimeter is indicated below each layout, with the red squares having exactly this visible perimeter. The stacking order is given by the numbers in the corners of each square. From left to right: optimal layout; optimal layout with a stacking order matching the $y$-order; optimal layout with a stacking order matching the inverse $y$-order; optimal staircase. These layouts were computed via (I)LPs using a difference of $0.001$ to turn strict inequalities into non-strict inequalities. \label{fig:difference}}
\end{figure}

Our results assume a strip of width at most 2, as well as distinct $y$-coordinates. While we can apply our constructions to the more general case, they may be of inferior quality. We leave to future work to establish whether and how these assumptions can be lifted and to determine the computational complexity of the problems in this area.

An interesting and practically relevant scenario for future work are rectangular symbols.
Our algorithms (constructions of layouts) can also be used for this case and yield results of high quality (see Figure~\ref{fig:vac2018}).
However, doing so loses the quality guarantees that we prove for the square case, since the resulting rectangle layouts will not optimize visible perimeter, but a variant of this measure in which horizontal and vertical visible edges have different weights.
Even more challenging are settings with differently sized symbols. We leave these question to future work.

\end{document}